\documentclass[conference,times,letterpaper]{IEEEtran}
\usepackage{times}                               % LaTeX 2e
\usepackage{verbatim}                            % LaTeX 2e
\usepackage{graphicx}
\usepackage{tabularx}
\usepackage{algo}
\usepackage{url}
\usepackage{amsmath}
\usepackage{amssymb}
\usepackage{amsfonts}
\usepackage{units}
\usepackage{textcomp,xspace}
\usepackage{enumerate}
\usepackage[ruled, vlined]{algorithm2e}
\usepackage{multirow, pgf}
\usepackage{subfigure}
\usepackage{xcolor}
\usepackage{authblk}

\usepackage{enumitem}

\newcommand{\tc}[1]{\ensuremath{\mathcal{M}({#1})}}

\newcommand{\mt}[1]{\ensuremath{\colon \mathbb{F}^{#1}_{2} \mapsto \mathbb{F}_2}}

\newtheorem{theorem}{Theorem}[section]
\newtheorem{lemma}[theorem]{Lemma}

\newtheorem{definition}[theorem]{Definition}

\input{Qcircuit}

\begin{document}

\title{Reversible Logic Circuit Complexity Analysis via Functional Decomposition}
\author{Anupam Chattopadhyay}
\author{Anubhab Baksi}
\affil{School of Computer Engineering, Nanyang Technological University, Singapore}

\maketitle
%%%%%%%%%%%%%%%%%%%%%%%%%%%%%%%%%%%%%%%%%%%%%%%%%%%%%%%%%%%%%%%%%%%%%%%%%%%%%%%%%%%%%%%%%%%%%%%%%%%%%%%%%%%%%%%%%%%%%%%%%%%%%%%
\begin{abstract}
Reversible computation is gaining increasing relevance in the context of several post-CMOS technologies, the most prominent of those being Quantum computing. One of the key theoretical problem pertaining to reversible logic synthesis is the upper bound of the gate count. Compared to the known bounds, the results obtained by optimal synthesis methods are significantly less. In this paper, we connect this problem with the multiplicative complexity analysis of classical Boolean functions. We explore the possibility of relaxing the ancilla and if that approach makes the upper bound tighter. Our results are negative. The ancilla-free synthesis methods by using transformations and by starting from an Exclusive Sum-of-Product (ESOP) formulation remain, theoretically, the synthesis methods for achieving least gate count for the cases where the number of variables $n$ is $ < 8$ and otherwise, respectively.
\end{abstract}

%%%%%%%%%%%%%%%%%%%%%%%%%%%%%%%%%%%%%%%%%%%%%%%%%%%%%%%%%%%%%%%%%%%%%%%%%%%%%%%%%%%%%%%%%%%%%%%%%%%%%%%%%%%%%%%%%%%%%%%%%%%%%%%

\section{Introduction}
A Boolean function $f$ is of the form $f:\{0,1\}^{n}\rightarrow\{0,1\}$ (or equivalently $f:\mathbb{V}_{2}^{n}\rightarrow \mathbb{V}_{2}$). The output of the Boolean function $f$ can be represented as a string $s$ of ones and zeros. It can also be represented as a multivariate polynomial over $GF(2)$. This polynomial can be expressed as an exclusive disjunction ($\oplus$) of a constant $a_{0}$ and one or more conjunctions of the function argument. This is called the Exclusive Sum-Of-Product (ESOP) representation. An alternative representation of the ESOP form is known as the Algebraic Normal Form (ANF). The general ANF for a function $f(x_{1},..,x_{n})$ over $n$-variables can be written as,
\begin{equation*}
\begin{split}
f(x_{1},..,x_{n}) = & a_{0} \oplus a_{1}x_{1} \oplus \cdots \oplus a_{i}x_{i} \oplus \cdots \oplus a_{n}x_{n} \\
& \oplus \cdots \oplus a_{1,2,...,n}x_{1}x_{2}\cdots x_{n}\\
\end{split}
\end{equation*}

\subsubsection*{Reversible and Irreversible Boolean functions}
An $n$-variable vectorial Boolean function is termed \textit{reversible} if all its output patterns map uniquely to an input pattern and vice-versa. It can be expressed as an $n$-input, $n$-output bijection or alternatively, as a Boolean permutation function over the truth value set $\{0, 1, \ldots, 2^{n-1}\}$. An \textit{irreversible} Boolean function
$f_{irr}:\{0,1\}^{n}\rightarrow\{0,1\}^{m}$ with $n \ne m$ can also be made reversible with the help of additional output lines, which adds distinguishing patterns to the irreversible output. Correspondingly, additional inputs are added. If an input line is constant-initialized and the constant is recovered after the circuit execution then, it is termed as \emph{ancilla}. Otherwise, it is termed as \emph{garbage}.

%%%%%%%%%%%%%%%%%%%%%%%%%%%%%%%%%%%%%%%%%%%%%%%%%%%%%%%%%%%%%

\subsection{Reversible Logic Synthesis} \label{sec:rev_synth}
Reversible Boolean logic synthesis is achieved with the help of reversible logic gates. The gates are characterized by their implementation cost in quantum technologies, which is denoted as the Quantum Cost (QC). A reversible gate library is a complete set of reversible gates which can be used to implement any reversible circuit. For example the set of NOT, CNOT, controlled-$V$ and Controlled-$V^+$, known as NCV, is a reversible gate library widely used in the literature. Recently, there has been a significant research activity towards the realization of quantum circuits using Clifford+T gates, considering the importance of fault tolerance in quantum computing. Efficient synthesis for NCV circuits~\cite{mmd_swap}, Clifford+T circuits~\cite{amy_middle} and mapping of NCV circuits to Clifford+T gates~\cite{ncv_to_clifford} have been proposed in the literature. Few gates from these libraries are outlined below. For detailed discussion on primitive quantum gates and their universality, readers may refer to~\cite{elem_gates, barenco_2bit}.

\begin{itemize}[leftmargin=*]
 \item \textbf{NOT gate}: This is represented using the matrix $\bigl(\begin{smallmatrix} 0&1 \\ 1&0 \end{smallmatrix}\bigr)$.
 \item \textbf{CNOT gate}: CNOT$(a,b)$=$(a,a\oplus b)$. This gate can be generalized with $Tof_n$ gate, where first $n-1$ variables are used as control lines. Generalized Toffoli gates form an universal reversible logic gate library, termed as MCT gate library.
 \item \textbf{Hadamard gate}: The unitary transformation for Hadamard gate is, $\frac{1}{\sqrt{2}}\bigl(\begin{smallmatrix} 1&1 \\ 1&-1 \end{smallmatrix}\bigr)$.
 \item \textbf{T gate}: The unitary transformation for this gate is, $\bigl(\begin{smallmatrix} 1&0 \\ 0&e^{\frac{i\pi}{4}} \end{smallmatrix}\bigr)$.
 \item \textbf{Phase gate}: Also denoted by $S$ and performs the unitary transformation of $\bigl(\begin{smallmatrix} 1&0 \\ 0&i \end{smallmatrix}\bigr)$. NOT, CNOT, Hadamard, $S$ and $S^+$ gates together form the Clifford gate family.
 \item \textbf{Toffoli gate}/\textbf{CCNOT gate}: The gate flips the target if both the control lines are set: $Tof(a,b,c) = (a,b,c\oplus ab)$.
\end{itemize}

\subsection{Cost Models} \label{sec:cost}
For evaluating the performance of the synthesis tools and benchmark circuit implementations, different cost models have been proposed in the literature. The most basic cost model is the number of reversible logic gates needed for the implementation. These logic gates could be, with rather large number of control lines, such as, for Multiple-Control Toffoli gates (MCT) or Mixed-Polarity Multiple Control Toffoli (MPMCT) gates. To closely reflect the cost in terms of quantum gates, the QC value is computed for each of these gates, which is nothing but the number of 2-qubit gates~\cite{maslov_web} needed to implement these circuits. In recent fault-tolerant Quantum circuit implementations, the cost is estimated in terms of T-count, corresponding to the Cliffort+T realization. Logical depth or T-depth.

Finally, considering the cost and difficulty of implementing large number of qubits, one performance indicator is the total number of lines or ancilla/garbage count. It is well-known that there exists a trade-off between gate count/QC and ancilla count~\cite{integration_tradeoff,anupam_tradeoff}. It is trivial to show that with increasing ancilla count, MCT gates with smaller control lines can be realized, thereby establishing a clear trade-off between QC and ancilla. However, it not clear if such a trade-off exists for gate count vs ancilla.

%%%%%%%%%%%%%%%%%%%%%%%%%%%%%%%%%%%%%%%%%%%%%%%%%%%%%%%%%%%%%%%%%%%%%%%%%%%%%%%%%%%%%%%%%%%%%%%%%%%%%%%%%%%%%%%%%%%%%%%%%%%%%%%%%%%%%%%%%%%%

\section{Background} \label{sec:existing_bounds}
For theoretical as well as practical interests, several researchers worked towards establishing the bounds on the gate count for reversible circuits. For selective classes of Boolean functions, much lower upper bound compared to generic Boolean functions has been reported in~\cite{constructive}. For reversible circuits of practical interest, e.g., modulo-exponentiation, addition, reversible circuit implementation with parameterized upper bounds have been established~\cite{takahashi_adder,markov_survey}. We outline the major results in the gate count upper bound for different synthesis methods.

\subsection{Upper Bound for Transformation-based Synthesis Methods} \label{sec:mmd_bound}
The gate cost upper bound for transformation-based synthesis methods, for $n$-variable reversible logic circuits, have been reported as following for two different gate libraries.
\begin{lemma}
The upper bound of gate count for MMD~\cite{mmd} using MPMCT library is $(n-1)2^n + 1$~\cite{mmd_tcad,mmd_swap}.
\end{lemma}

\begin{lemma}
The upper bound of gate count for MMD~\cite{mmd} using MPMCT(+Fredkin gate) library is $(n-2)2^n + 2 + n$~\cite{mmd_swap}.
\end{lemma}

\subsection{Upper Bound for BDD-based Synthesis Methods} \label{sec:bdd_bound}
The Binary Decision Diagram (BDD)-based synthesis method, originally proposed in~\cite{wille_bdd}, observed an gate count upper bound of $3.2^n$. Note that this method generates high number of ancilla, which is in the order of $O(2^n)$. By utilizing a bi-conditional decomposition, for selective Boolean functions, a tighter upper bound is proposed in~\cite{bbdd}, however, the worst-case upper bound remains the same for generic Boolean functions. Ancilla-free reversible logic synthesis for BDDs has been proposed recently~\cite{mathias_ancilla_free_bdd}, which does not report any theoretical bounds so far.

\subsection{Upper Bound for ESOP-based Synthesis Methods} \label{sec:esop_bound}
The tightest upper bound for reversible circuits have been reported in~\cite{nabilla_bound} following a mix of Young subgroups and Exclusive Sum-Of-Products (ESOP)-based synthesis approach. In the first phase of the synthesis with young subgroups, it is established that an $n$-variable reversible function can be implemented with a circuit of $(2n-1)$ single target gates. Further, it was shown that, an $n$-variable single target gate can be realized with at most $29\cdot2^{n-8}$ MPMCT gates, if $n \geq 8$. Therefore, the gate count upper bound for the reversible circuit is $29\cdot2^{n-8}\cdot(2n-1)$ MPMCT gates, if $n \geq 8$. The upper bound of the single target gate implementation is established by using a result from~\cite{gaidukov}. Note that this method does not require any ancilla line.

For the rest of the paper, we will closely follow the above synthesis method and focus on the complexity of a single-output Boolean function. We will assume the single target distribution obtained from Young subgroup method~\cite{nabilla_bound} and concentrate on the functional decomposition approach of Boyar et al~\cite{boyar_bool} to check for potential improvements.

\subsection{Multiplicative Complexity} \label{sec:mult_complexity}
In parallel with the efforts for enumerating the reversible circuit complexity, progress in the analysis of classical Boolean functions in terms of multiplicative complexity is reported~\cite{boyar_bool}. The primary application for this analysis has been towards the design of cryptographic primitives, though.

\begin{definition}
The multiplicative complexity $C_\wedge(f)$ of a Boolean function is the minimum number of multiplications (AND-gates) that are sufficient to evaluate the function over the basis $(\wedge, \oplus, \neg)$.
\end{definition}

It has been established that the multiplicative complexity of functions having degree $d$ is at least $d-1$~\cite{schnorr_mult_c} and the multiplicative complexity of a randomly selected $n$-variable Boolean function is at most $2^{\frac{n}{2} + 1} - n/2 -2$ ($n$ even)~\cite{boyar_bool}.

Intuitively, multiplicative complexity of a Boolean function has some correspondence with the Toffoli gate. The number of AND gates in a single term of ESOP expression represents the number of control lines in a Toffoli gate. In the following sections, we follow the same approach of the derivation of multiplicative complexity and explore if it reduces the theoretical upper bound of gate count compared to the currently known bounds.

\section{Upper Bounds via Multiplicative Complexity Analysis} \label{sec:tighter_bound}
In this section, we first revisit the function decomposition approach used in~\cite[Section 4]{nabilla_bound} from a mathematical perspective. We show that their work can be described as solving a linear recurrence of the form: $$f(n) = 2f(n-1) + 2$$ with the initial conditions $f(5) = 10$, $f(4) = 4$.

The reason for adopting the recurrence relation, $f(n) = 2f(n-1) + 2$ is described in second paragraph of the proof of \cite[Theorem 1]{nabilla_bound}, so we skip that part. The case for $f(5)$ is described in~\cite{wille_sat}, which is cited in~\cite{nabilla_bound}.

Now, assume that, we want to solve: $f(n) = 2f(n-1) + 2$.
We substitute $n$ by $n+1$ to yield, $f(n+1) = 2f(n) + 2$.
Now, subtracting, we get $f(n+1) - f(n) = 2f(n) -2f(n-1)$ $\implies f(n+1) = 3f(n) - 2f(n-1)$.

Hence, the characteristic equation is: $x^2=3x-2 \implies (x-2)(x-1) =0 \implies x=1,2$, and hence, $f(n) = \alpha\cdot1^n+\beta\cdot2^n$, where $\alpha, \beta$ are constants to be determined from the initial conditions.

Now, given that, $f(5)=10, f(4)=4$, %(we need to know at least two initial values to solve a second order recurrence relation)
we generate the linear equations: $f(5)= \alpha + 32\beta = 10, f(4)=\alpha+16\beta=4$, which gives, $\alpha=-2, \beta=\frac38=3\cdot2^{-3}$.

Therefore, $f(n) = -2+3\cdot2^{-3}\cdot2^n = 3\cdot2^{n-3}-2$.

\subsection{Based on Function Decomposition}
A recursive decomposition procedure can be used for systematic implementation of reversible Boolean functions. From a given ESOP form, three following functional decompositions can be applied.
\begin{equation}
f = \overline{x_i}\cdot f_{x_i=0} \oplus x_i\cdot f_{x_i=1} \label{eqn-shannon}
\end{equation}
\begin{equation}
f = f_{x_i=0} \oplus x_i\cdot f_{x_i=2} \label{eqn-pdavio}
\end{equation}
\begin{equation}
f = f_{x_i=1} \oplus \overline{x_i}\cdot f_{x_i=2} \label{eqn-ndavio}
\end{equation}
where $f_{x_i=2} = f_{x_i=0} \oplus f_{x_i=1}$. These are known as Shannon, positive Davio and negative Davio decompositions respectively. Here, we adopt the positive Davio decomposition. The Shannon decomposition leads to a $2$-MCT gate implementation, whereas both positive Davio and negative Davio can be realized with a single MCT or MPMCT gate respectively. Given a Boolean function $f \mt{n}$ in its ESOP representation, we implement ($\wedge$ and $\oplus$ gates) it with MCT gate library only. We denote the MCT gate complexity of $f$ by \tc{f}.

In~\cite[Lemma 5]{boyar_bool}, it was proved that it requires at most $2^n - n - 1$ AND ($\wedge$) gates for realizing all the positive minterms of $n$ Boolean variables. In that, authors proposed a circuit of alternating $\wedge$ and $\oplus$ gates. It is evident that this set of all possible minterms can also be implemented with only MCT gates, when $0$-initialized ancilla lines are used. A representative implementation is shown in the following Fig.~\ref{minterm}. This requires $2^n - n - 1$ MCT gates and total $2^n-1$ lines.

\begin{figure}[htb]
\centering
$
\Qcircuit @C=0.7em @R=.5em {
\lstick{x_1} & \ctrl{1} & \ctrl{2}   & \qw      & \ctrl{1} & \qw & \rstick{x_1}\\
\lstick{x_2} & \ctrl{2} & \qw        & \ctrl{1} & \ctrl{1} & \qw & \rstick{x_2} \\
\lstick{x_3} & \qw      & \ctrl{2}   & \ctrl{3} & \ctrl{4} & \qw & \rstick{x_3} \\
\lstick{0}   & \targ    & \qw        & \qw      & \qw      & \qw & \rstick{f = x_1x_2}\\
\lstick{0}   & \qw      & \targ      & \qw      & \qw      & \qw & \rstick{f = x_1x_3}\\
\lstick{0}   & \qw      & \qw        & \targ    & \qw      & \qw & \rstick{f = x_2x_3}\\
\lstick{0}   & \qw      & \qw        & \qw      & \targ    & \qw & \rstick{f = x_1x_2x_3}
}
$
\caption{Implementation of Minterms}\label{minterm}
\end{figure}
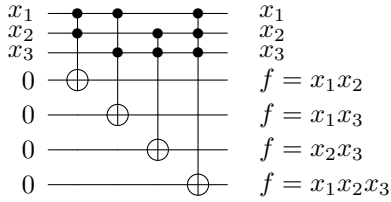

\begin{lemma} \label{lemma_xor}	
If all the positive minterms of $n$ variables are computed, then, for realizing $r$ different $n$-variable Boolean functions over the basis $(\wedge, \oplus, 1)$, at most $r(\frac{2^n-1}{2} - 1)$ MCT gates are needed.
\end{lemma}

\begin{proof}
Every function $\{f_i | 1 \leq i \leq r\}$ from the above set of $r$ functions can use a set of positive minterms to prevent formation of sharing co-factors. To begin with, we assume that the MCT of all minterms (except constant function) is computed, which requires $2^n-2$ MCT gates. Further, the minterm $0$ is made available as an ancilla line. Therefore, for computing each function, we need to look into the complete set and the zero set. In either case, the target function will never require more than half of the minterms to undergo exclusive-OR operation. However, if we assume that $f_1$ to need half of the minterms and $f_2$ to need the other half, there is a possibility that $f_3$ will require a few more minterms from the $f_1$ and a few minterms from $f_2$. In that way, storing the intermediate functions will not be useful and we get an upper bound of $(\frac{2^n-1}{2} - 1)$ MCT gates for each of the functions.
\end{proof}

We start with an idea similar to~\cite[Theorem 6]{boyar_bool}. Following positive Davio decomposition, $\forall f \mt{n}$, it is possible to find $f_1,f_{0} \mt{n-1}$ such that, $$ f(x_1,x_2,\ldots,x_n) = x_1f_1(x_2,\ldots,x_n) \oplus f_{0}(x_2,\ldots,x_n).$$ Alternatively, it can be stated that, $f$ is divided by $x_1$, with $f_1$ being the quotient and $f_{0}$ being the remainder. For example,  if $f(x_1,x_2,x_3) = x_1x_2x_3 \oplus x_1x_2 \oplus x_2x_3 \oplus x_1 \oplus x_2\oplus 1$, then $f_1(x_2,x_3)=x_2x_3\oplus x_2\oplus1$ and $f_{0}(x_2,x_3) = x_2x_3\oplus x_2\oplus1$. Also, notice that both $f_0, f_1$ are of (at most) $n-1$ variables.

Note that, if we can implement $f_1$ and $f_0$, then we can readily implement $f$ just by using one MCT gate (using $x_1, f_1$ as the control lines and $f_0$ as the target line, as depicted in Fig.~\ref{stp-1}).

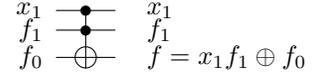
\begin{figure}[htb]
\centering
$
\Qcircuit @C=0.7em @R=.5em {
\lstick{x_1} & \ctrl{1} & \qw & \rstick{x_1}\\
\lstick{f_1} & \ctrl{1} & \qw & \rstick{f_1} \\
%& \ctrlo{0} & \qw
\lstick{f_0} & \targ & \qw & \rstick{f = x_1f_1 \oplus f_0}
}
$
\caption{Implementation of $f$ from $f_1$ and $f_0$ as inputs}\label{stp-1}
\end{figure}

Thus, we make the following observation: $$\tc f \leq 1 + \tc{\{f_1, f_0\}},$$ where \tc{\{f_1,f_0\}} denotes the MCT gate complexity of $f_1,f_0$, combined. Therefore, we can express the MCT gate complexity of an $n$-input function $f$ in terms functions of lower number of input variables. We can recursively use this decomposition:
\begin{align*}
{f} = &{x_1f_1\oplus f_{0}} \\
    = &{x_1(x_2f_{11}\oplus f_{10})\oplus (x_2f_{01}\oplus f_{00})} \\
    = &{x_1(x_2(x_3f_{111}\oplus f_{110})\oplus x_3f_{101}\oplus f_{100})} \\
      & \oplus {(x_2(x_3f_{011}\oplus f_{010})\oplus x_3f_{001}\oplus f_{000})}
\end{align*}
It can be noted that if we conduct this decomposition in steps of $4$ factors, instead of steps of $2$, exactly the same number of gates would be needed at the end. This is due to the fact that the $4$-factor decomposition, if implemented in reversible circuit, would need $3$ MCT gates.

In the above decomposition, at $k$-th step, we require $2^k-1$ additional MCT gates other than the ones needed by the new functions of lower input variables. Thus, after $k$-th step, we can write:
\begin{align}
&\tc{f} \notag\\& \leq 1 + \tc{\{f_1,f_0\}} \notag\\ & \leq 3 + \tc{\{f_{11}, f_{10}, f_{01},f_{00} \}} \notag
\\ & \hspace{.29\textwidth}\vdots  \notag\\ & \leq (2^k-1) + \tc{\{f_{i_1,i_2,\cdots,i_k}|{i_1,i_2,\ldots,i_k}\in \{0,1\}^k\}} \label{eqn-1}
\end{align}

Note that, all functions $\{f_{i_1,i_2,\cdots,i_k}|$ ${i_i,i_2,\ldots,i_k}\in$ $\{0,1\}^k\}$ are (at most) of $n-k$ variables, hence, $\{f_{i_1,i_2,\cdots,i_k}|$ ${i_1,i_2,\ldots,i_k}\in$ $\{0,1\}^k\}$ denotes some circuit with $n-k$ variables.

Therefore, using the lemma~\ref{lemma_xor} (here, $r = 2^k$) and the MCT gate count for the minterm constructions, in equation~\ref{eqn-1}, we can write
\begin{align}
& \tc{f} \notag\\
%& \leq 1 + \tc{\{f_1,f_0\}} \notag\\
%& \leq 3 + \tc{\{f_{11}, f_{10}, f_{01},f_{00} \}} \notag
%\\
%& \hspace{.29\textwidth}\vdots  \notag\\
& \leq (2^k-1) + \tc{\{f_{i_1,i_2,\cdots,i_k}|{i_1,i_2,\ldots,i_k}\in \{0,1\}^k\}} \notag\\
& \leq (2^k-1) + \{2^{n-k} - (n-k) - 1\} + 2^k\{\frac{2^{n-k}-1}{2} - 1\} \notag\\
& \leq (2^k-1) + \{2^{n-k} - (n-k) - 1\} + 2^{n-1} - 3\cdot2^{k-1}
\label{eqn-AC}
\end{align}

Assuming $k=n/2$ to be the point, where we stop the functional decomposition and construct the functions from the available positive minterms, we obtain
\begin{align}
\tc{f} \leq 2^{n-1} + 2^{n/2-1} - n/2 - 2
\label{eqn-finalonegate}
\end{align}

This leads to an overall MCT gate complexity of
\begin{align}
(2n-1)(2^{n-1} + 2^{n/2-1} - n/2 - 2)
\label{eqn-finalMCT}
\end{align}
by following the Young subgroup decomposition. Corresponding bound for the garbage count is $O(2^{n/2})$, which is introduced for the minterm construction and for the realization of the functions once the decomposition procedure is stopped.

\begin{figure}[htbp]
\resizebox{0.9\columnwidth}{!}{
\centering
\includegraphics[width=\columnwidth]{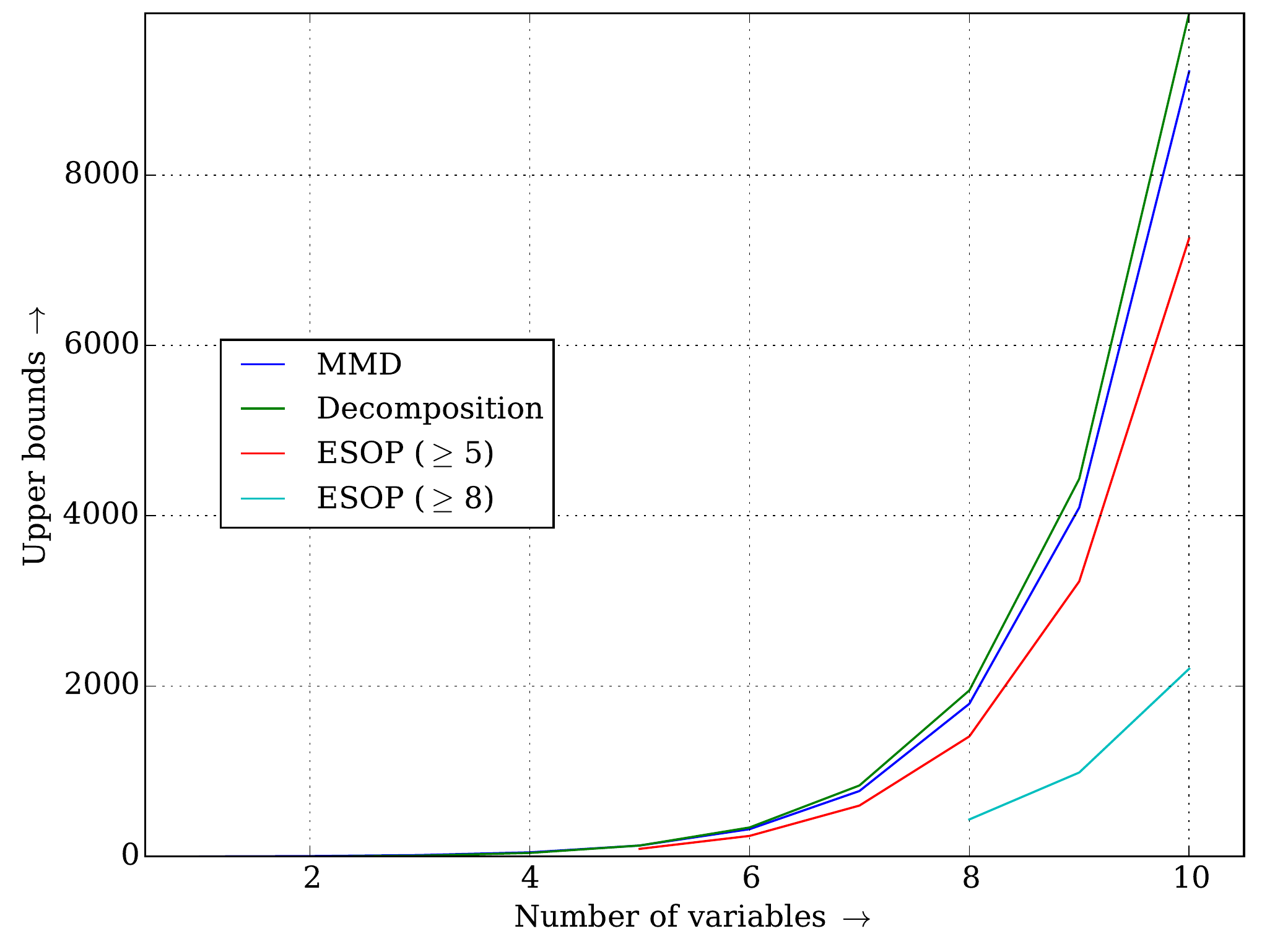}
}
\caption{Upper Bound of MCT Gate Complexity}\label{fig:complexity}
\vspace{-0.3cm}
\end{figure}

\subsection{Discussion}
It can be noted that the MCT gate complexity, derived through the approach of~\cite{boyar_bool} is worse compared to that of the transformation-based method (for $n<8$) and the ESOP-based method (for $n\geq 8$) (see Fig. \ref{fig:complexity}). This leads to the conclusion that introducing ancilla/garbage lines does not lead to any reduction in the gate count complexity, if functional decomposition approach is followed. On the other hand, different complexity bounds in different synthesis approaches indicate that it might be beneficial to adopt a hybrid synthesis approach for reversible circuit realization.

%{\color{blue} In~\cite{maslov_type_complexity}, it is conjectured that the Toffoli gate count for an $n$-variable reversible function is $\lesssim \sqrt{n}2^{n/2}$. Note that, once all the minterms are computed, the derivation of the functions require only CNOT gates. On that basis, this conjecture can be proved by only enumerating the Toffoli gates from the eqn.~\ref{eqn-AC}, as following. We denote the count of Toffoli gates for the Boolean function $f$ by $\tofc{f}$.}
%
%\begin{align}
%& \tofc{f} \notag\\
%& \leq (2^k-1) + \tc{\{f_{i_1,i_2,\cdots,i_k}|{i_1,i_2,\ldots,i_k}\in \{0,1\}^k\}} \notag\\
%& \leq (2^k-1) + \{2^{n-k} - (n-k) - 1\}\notag\\
%\label{eqn-Tof1}
%\end{align}
%
%By choosing $k$ as $\frac{n+\log_2{n}}{2}$, we obtain.
%
%\begin{align}
%& \tofc{f} \notag\\
%& \leq (2^{\frac{n+\log_2{n}}{2}}-1) + \{2^{\frac{n-\log_2{n}}{2}} - (n-\log_2{n})/2 - 1\}\notag\\
%& \leq 2^{\frac{n+\log_2{n}}{2}} + 2^{\frac{n-\log_2{n}}{2}} - (n-\log_2{n}-1)/2\notag\\
%& \leq 2^{\frac{n+\log_2{n}}{2}} + 2^{\frac{n-\log_2{n}}{2}} - (n-\log_2{n}-1)/2\notag\\
%& \lesssim \sqrt{n}2^{n/2} + \mathcal{O}(\sqrt{n}2^{n/2})\notag\\
%\label{eqn-Tof2}
%\end{align}

%%%%%%%%%%%%%%%%%%%%%%%%%%%%%%%%%%%%%%%%%%%%%%%%%%%%%%%%%%%%%%%%%%%%%%%%%%%%%%%%%%%%%%%%%%%%%%%%%%%%%%%%%%%%%%%%%%%%%%%%%%%%%%%%
\bibliographystyle{IEEEtran}
\bibliography{rev}

% Generated by IEEEtran.bst, version: 1.14 (2015/08/26)
\begin{thebibliography}{10}
\providecommand{\url}[1]{#1}
\csname url@samestyle\endcsname
\providecommand{\newblock}{\relax}
\providecommand{\bibinfo}[2]{#2}
\providecommand{\BIBentrySTDinterwordspacing}{\spaceskip=0pt\relax}
\providecommand{\BIBentryALTinterwordstretchfactor}{4}
\providecommand{\BIBentryALTinterwordspacing}{\spaceskip=\fontdimen2\font plus
\BIBentryALTinterwordstretchfactor\fontdimen3\font minus
  \fontdimen4\font\relax}
\providecommand{\BIBforeignlanguage}[2]{{%
\expandafter\ifx\csname l@#1\endcsname\relax
\typeout{** WARNING: IEEEtran.bst: No hyphenation pattern has been}%
\typeout{** loaded for the language `#1'. Using the pattern for}%
\typeout{** the default language instead.}%
\else
\language=\csname l@#1\endcsname
\fi
#2}}
\providecommand{\BIBdecl}{\relax}
\BIBdecl

\bibitem{mmd_swap}
M.~Soeken and A.~Chattopadhyay, ``Fredkin-enabled transformation-based
  reversible logic synthesis,'' in \emph{Multiple-Valued Logic (ISMVL), 2015
  IEEE International Symposium on}, May 2015, pp. 60--65.

\bibitem{amy_middle}
M.~Amy, D.~Maslov, M.~Mosca, and M.~Roetteler, ``A meet-in-the-middle algorithm
  for fast synthesis of depth-optimal quantum circuits,'' \emph{Computer-Aided
  Design of Integrated Circuits and Systems, IEEE Transactions on}, vol.~32,
  no.~6, pp. 818--830, June 2013.

\bibitem{ncv_to_clifford}
\BIBentryALTinterwordspacing
D.~Miller, M.~Soeken, and R.~Drechsler, ``\BIBforeignlanguage{English}{Mapping
  ncv circuits to optimized clifford+t circuits},'' in
  \emph{\BIBforeignlanguage{English}{Reversible Computation}}, ser. Lecture
  Notes in Computer Science, S.~Yamashita and S.-i. Minato, Eds.\hskip 1em plus
  0.5em minus 0.4em\relax Springer International Publishing, 2014, vol. 8507,
  pp. 163--175. [Online]. Available:
  \url{http://dx.doi.org/10.1007/978-3-319-08494-7_13}
\BIBentrySTDinterwordspacing

\bibitem{elem_gates}
\BIBentryALTinterwordspacing
A.~Barenco, C.~H. Bennett, R.~Cleve, D.~P. DiVincenzo, N.~Margolus, P.~Shor,
  T.~Sleator, J.~A. Smolin, and H.~Weinfurter, ``Elementary gates for quantum
  computation,'' \emph{Phys. Rev. A}, vol.~52, pp. 3457--3467, Nov 1995.
  [Online]. Available: \url{http://link.aps.org/doi/10.1103/PhysRevA.52.3457}
\BIBentrySTDinterwordspacing

\bibitem{barenco_2bit}
A.~Barenco, ``\BIBforeignlanguage{English}{A universal two-bit gate for quantum
  computation},'' \emph{\BIBforeignlanguage{English}{Proceedings: Mathematical
  and Physical Sciences}}, vol. 449, no. 1937, pp. pp. 679--683, 1995.

\bibitem{maslov_web}
``{Reversible Logic Synthesis Benchmarks Page}, howpublished =
  {\url{http://webhome.cs.uvic.ca/~dmaslov/}}, note = {Accessed: 2015-10-19}.''

\bibitem{integration_tradeoff}
\BIBentryALTinterwordspacing
R.~Wille, M.~Soeken, D.~M. Miller, and R.~Drechsler, ``Trading off circuit
  lines and gate costs in the synthesis of reversible logic,''
  \emph{Integration, the VLSI Journal}, vol.~47, no.~2, pp. 284 -- 294, 2014.
  [Online]. Available:
  \url{http://www.sciencedirect.com/science/article/pii/S0167926013000436}
\BIBentrySTDinterwordspacing

\bibitem{anupam_tradeoff}
\BIBentryALTinterwordspacing
A.~Chattopadhyay, N.~Pal, and S.~Majumder, ``Ancilla-quantum cost trade-off
  during reversible logic synthesis using exclusive sum-of-products,''
  \emph{CoRR}, vol. abs/1405.6073, 2014. [Online]. Available:
  \url{http://arxiv.org/abs/1405.6073}
\BIBentrySTDinterwordspacing

\bibitem{constructive}
\BIBentryALTinterwordspacing
A.~Chattopadhyay, S.~Majumder, C.~Chandak, and N.~Chowdhury, ``Constructive
  reversible logic synthesis for boolean functions with special properties,''
  in \emph{Reversible Computation}, ser. Lecture Notes in Computer Science,
  S.~Yamashita and S.-i. Minato, Eds.\hskip 1em plus 0.5em minus 0.4em\relax
  Springer International Publishing, 2014, vol. 8507, pp. 95--110. [Online].
  Available: \url{http://dx.doi.org/10.1007/978-3-319-08494-7_8}
\BIBentrySTDinterwordspacing

\bibitem{takahashi_adder}
\BIBentryALTinterwordspacing
Y.~Takahashi, S.~Tani, and N.~Kunihiro, ``Quantum addition circuits and
  unbounded fan-out,'' \emph{Quantum Info. Comput.}, vol.~10, no.~9, pp.
  872--890, Sep. 2010. [Online]. Available:
  \url{http://dl.acm.org/citation.cfm?id=2011464.2011476}
\BIBentrySTDinterwordspacing

\bibitem{markov_survey}
M.~Saeedi and I.~L. Markov, ``Synthesis and optimization of reversible
  circuits- a survey,'' \emph{ACM Comput. Surv.}, vol.~45, no.~2, pp.
  21:1--21:34, Mar. 2013.

\bibitem{mmd}
D.~Miller, D.~Maslov, and G.~Dueck, ``A transformation based algorithm for
  reversible logic synthesis,'' in \emph{Design Automation Conference}, 2003,
  pp. 318--323.

\bibitem{mmd_tcad}
D.~Maslov, G.~Dueck, and D.~Miller, ``Toffoli network synthesis with
  templates,'' \emph{Computer-Aided Design of Integrated Circuits and Systems,
  IEEE Transactions on}, vol.~24, no.~6, pp. 807--817, June 2005.

\bibitem{wille_bdd}
R.~Wille and R.~Drechsler, ``{BDD-based Synthesis of Reversible Logic for Large
  Functions},'' in \emph{Proceedings of the 46th Annual Design Automation
  Conference}, ser. DAC '09, 2009, pp. 270--275.

\bibitem{bbdd}
A.~Chattopadhyay, A.~Littarru, L.~Amaru, P.-E. Gaillardon, and G.~De~Micheli,
  ``Reversible logic synthesis via biconditional binary decision diagrams,'' in
  \emph{Multiple-Valued Logic (ISMVL), 2015 IEEE International Symposium on},
  May 2015, pp. 2--7.

\bibitem{mathias_ancilla_free_bdd}
\BIBentryALTinterwordspacing
M.~Soeken, L.~Tague, G.~W. Dueck, and R.~Drechsler, ``Ancilla-free synthesis of
  large reversible functions using binary decision diagrams,'' \emph{J. Symb.
  Comput.}, vol.~73, pp. 1--26, 2016. [Online]. Available:
  \url{http://dx.doi.org/10.1016/j.jsc.2015.03.002}
\BIBentrySTDinterwordspacing

\bibitem{nabilla_bound}
N.~Abdessaied, M.~Soeken, M.~K. Thomsen, and R.~Drechsler, ``Upper bounds for
  reversible circuits based on young subgroups,'' \emph{Information Processing
  Letters}, vol. 114, no.~6, pp. 282--286, 2014.

\bibitem{gaidukov}
A.Gaidukov, ``Algorithm to derive minimum esop for 6-variable function,'' in
  \emph{International Workshop on Boolean Problems}, 2002, pp. 141--148.

\bibitem{boyar_bool}
\BIBentryALTinterwordspacing
{Joan Boyar, Ren{\'{e}} Peralta, Denis Pochuev}, ``{On the multiplicative
  complexity of Boolean functions over the basis $(\wedge, \oplus, 1)$},''
  \emph{Theor. Comput. Sci.}, vol. 235, no.~1, pp. 43--57, 2000. [Online].
  Available: \url{http://dx.doi.org/10.1016/S0304-3975(99)00182-6}
\BIBentrySTDinterwordspacing

\bibitem{schnorr_mult_c}
\BIBentryALTinterwordspacing
C.-P. Schnorr, ``The multiplicative complexity of boolean functions,'' in
  \emph{Proceedings of the 6th International Conference, on Applied Algebra,
  Algebraic Algorithms and Error-Correcting Codes}, ser. AAECC-6.\hskip 1em
  plus 0.5em minus 0.4em\relax London, UK, UK: Springer-Verlag, 1989, pp.
  45--58. [Online]. Available:
  \url{http://dl.acm.org/citation.cfm?id=646025.676419}
\BIBentrySTDinterwordspacing

\bibitem{wille_sat}
D.~Grosse, R.~Wille, G.~Dueck, and R.~Drechsler, ``Exact multiple-control
  toffoli network synthesis with sat techniques,'' \emph{Computer-Aided Design
  of Integrated Circuits and Systems, IEEE Transactions on}, vol.~28, no.~5,
  pp. 703--715, May 2009.

\end{thebibliography}
\end{document}